\documentclass[10pt,journal]{IEEEtran}
\usepackage[colorlinks,
linkcolor=red,
anchorcolor=green,
citecolor=blue
]{hyperref} 
\usepackage{amsfonts,amssymb,amsmath,array}
\usepackage{latexsym}
\usepackage{CJK}
\usepackage{cite}
\usepackage{bm}
\usepackage{soul}

\usepackage{graphicx}
\usepackage{epstopdf}
\usepackage{epsfig}
\usepackage{subfigure} 

\usepackage{verbatim}  

\usepackage[usenames,dvipsnames]{color}
\definecolor{mycyan}{cmyk}{.3,0,0,0}
\usepackage{colortbl}
\usepackage{enumerate}
\usepackage{cases}

\usepackage{mathrsfs}
\usepackage{algorithm} 
\usepackage{algorithmic} 
\usepackage{booktabs}
\usepackage{textcomp}
\usepackage{multirow}
\usepackage{lettrine}
\usepackage[scaled=0.92]{helvet}

\usepackage{algorithm}
\usepackage{algorithmic}

\usepackage{amsthm} 

\newtheorem{thm}{Theorem}
\newtheorem{cor}{Corollary}
\newtheorem{defn}{Definition}
\newtheorem{rem}{Remark}

\begin{document}
\title{Differential Message Importance Measure: A New Approach to the Required Sampling Number in Big Data Structure Characterization}

\author{Shanyun~Liu, Rui~She, Pingyi~Fan,~\IEEEmembership{Senior Member, ~IEEE}\\


\thanks{
Shanyun~Liu, Rui~She and Pingyi Fan are with Tsinghua National Laboratory for Information Science and Technology(TNList) and the Department of Electronic Engineering,  Tsinghua University,  Beijing,  P.R. China, 100084. e-mail: \{liushany16,~sher15@mails.tsinghua.edu.cn, fpy@tsinghua.edu.cn.\} }
\thanks{This work was supported in part by the National Natural Science Foundation of China (NSFC) under Grant 61771283 and 61621091, and in part by the China Major State Basic Research Development Program (973 Program) under Grant 2012CB316100(2).}

}
\maketitle

\graphicspath{{Figures/}}

\begin{abstract}
Data collection is a fundamental problem in the scenario of big data, where the size of sampling sets plays a very important role, especially in the characterization of data structure. This paper considers the information collection process by taking message importance into account, and gives a distribution-free criterion to determine how many samples are required in big data structure characterization. Similar to differential entropy, we define differential message importance measure (DMIM) as a measure of message importance for continuous random variable. The DMIM for many common densities is discussed, and high-precision approximate values for normal distribution are given. Moreover, it is proved that the change of DMIM can describe the gap between the distribution of a set of sample values and a theoretical distribution. In fact, the deviation of DMIM is equivalent to Kolmogorov-Smirnov statistic, but it offers a new way to characterize the distribution goodness-of-fit. Numerical results show some basic properties of DMIM and the accuracy of the proposed approximate values. Furthermore, it is also obtained that the empirical distribution approaches the real distribution with decreasing of the DMIM deviation, which contributes to the selection of suitable sampling points in actual system.
\end{abstract}


\IEEEpeerreviewmaketitle

\section{Introduction}\label{Sec:Introduction}
The actual system of big data needs to process lots of data within a limited time generally, so many researches are on sample data to improve their efficiency \cite{chen2014big,tahmassebpour2017new}. In fact, sampling technology is intensely effective for solving the challenges in big data, such as intrusion detection \cite{meng2017enhancing} and privacy-preserving approximate search \cite{zhou2017hermes}. One basic problem that can occur with sampling is that how many samples is required to have a good characterization of the big data structure, e.g. fitting the real distribution. Too many samples means wasting of resources, while too little samples is along with great bias. Distribution goodness-of-fit is generally used to describe this problem, which focuses on the error magnitude between the distribution of a set of sample values and the real distribution, and it plays a fundamental role in signal processing and information theory. This paper desires to solve this problem based on information theory.

Shannon entropy \cite{shannon2001mathematical} is possibly the most important quantity in information theory, which describes the fundamental laws of data compression and communication \cite{verdu1998fifty}. Due to its success, numerous entropies have been provided in order to extend information theory. Among them, the most successful expansion is R{\'e}nyi entropy \cite{renyi1961measures}. There are many applications based on R{\'e}nyi entropy, such as hypothesis testing \cite{morales2000renyi,van2014renyi}.

Actually, entropy is a quantity with respect to probability distribution, which satisfies the intuitive notion of what a measure of information should be \cite{Elements}. Generally, the events are naturally endowed with importance label and the process of fitting is equivalent to the process of information collection. Therefore, in this paper, we propose differential message importance measure (DMIM) as a measure of information for continuous random variable to characterize the process of information collection. DMIM is expanded from discrete message importance measure (MIM) \cite{fan2016message} which is such an information quantity coming from the intuitive notion of information importance for small probability event. Much of research in the last two decades has examined the application of small probability event in big data \cite{ramaswamy2000efficient,julisch2002mining,zieba2015counterterrorism}. Recent studies also show that MIM has many applications in big data, such as information divergence measures \cite{she2017amplifying} and compressed data storage \cite{liu2017non}.

Much of the research in the goodness of fit in the past several decades focused on the Kolmogorov-Smirnov test \cite{massey1951kolmogorov,lilliefors1967kolmogorov}. Based on it, \cite{Resnick1992Advantures} gave an
error estimation of empirical distribution. \cite{song2002goodness} presented a general method for distribution-free goodness-of-fit tests based on Kullback-Leibler discrimination information. The problem of testing goodness-of-fit in a discrete setting was discussed in \cite{harremoes2012information}. All these result can describe the goodness of fit very well and guide us to choose the sampling numbers. However, they all consider this problem based on the divergence of two distributions, so the previous results can not describe the message carried by each sample and the information change with the increase of the sampling size, which means that they can not visually display the process of information collection. In fact, DMIM is the proper measure to help us consider the problem of goodness-of-fit in the view of the information collection of continuous random variables. Moreover, Compared with Kolmogorov-Smirnov statistic, DMIM also shows the relationship between the variance of a random variable and the error estimation of empirical distribution.

The rest of this paper is organized as follows. Section \ref{sec:two} introduces the definition and the relationship between MIM and DMIM. In Section \ref{sec:three}, the properties of DMIM are introduced. Then, the DMIM of some basic continuous distributions are discussed in Section \ref{sec:four}, in which we give the asymptotic analysis of normal distribution. In Section \ref{sec:five}, the goodness of fit with DMIM is presented in order to analyze the process of information collection. The validity of proposed theoretical results is verified by the simulation results in Section \ref{sec:numerical}. Finally, we finish the paper with conclusions in Section \ref{sec:seven}.

\section{The Definition of DMIM}\label{sec:two}
\subsection{Differential Message Important Measure}
 \begin{defn}\label{Non-parametric MIM}
The DMIM $l(X)$ of a continuous random variable $X$ with density $f(x)$ is defined as
\begin{equation}
l(X) = {\int_S {f(x){e^{ - f(x)}}dx} },
\end{equation}
where $S$ is the support set of the random variable.
\end{defn}

For most continuous random variables, the DMIM has no simple expression and the integral form is inconvenient for numerical calculation, so we will give another form of it.
\begin{thm}\label{thm:lf}
The DMIM of a continuous random variable $X$ with density $f(x)$ can be written as
\begin{equation}\label{equ:thm lf}
 l(X)= 1 + \sum\limits_{n = 1}^\infty  {{{{{\left( { - 1} \right)}^n}} \over {n!}}\int_{ - \infty }^{ + \infty } {{{\left( {f (x)} \right)}^{n + 1}}dx} }.
\end{equation}
\end{thm}
\begin{proof}
In fact, we obtain
\begin{flalign}\label{equ:Normal}
l(X)=&\int_{ - \infty }^{ + \infty } {f (x)} {e^{ - f (x)}}dx \\
=& \int_{ - \infty }^{ + \infty } {f (x)} \sum\limits_{ n= 0}^\infty  {{{{{\left( { - f (x)} \right)}^n}} \over {n!}}} dx\tag{\theequation a}\label{equ:Normal a}\\
 = &\int_{ - \infty }^{ + \infty } {\sum\limits_{n = 0}^\infty  {{{\left( { - 1} \right)}^n}{{{{\left( {f (x)} \right)}^{n + 1}}} \over {n!}}} dx}  \tag{\theequation b}\label{equ:Normal b} \\
  =& \int_{ - \infty }^{ + \infty } {f (x)dx + \sum\limits_{n = 1}^\infty  {\int_{ - \infty }^{ + \infty } {{{\left( { - 1} \right)}^n}{{{{\left( {f (x)} \right)}^{n + 1}}} \over {n!}}dx} } }\tag{\theequation c}\label{equ:Normal c}\\
   =& 1 + \sum\limits_{n = 1}^\infty  {{{{{\left( { - 1} \right)}^n}} \over {n!}}\int_{ - \infty }^{ + \infty } {{{\left( {f (x)} \right)}^{n + 1}}dx} }. \tag{\theequation d}\label{equ:Normal d}
\end{flalign}
\end{proof}

\subsection{Relation of DMIM to MIM}
For a random variable $X$ with density $f(x)$, we divide the range of $X$ into bins of length $\Delta$. We also suppose that $f(x)$ is continuous within the bins. According to the mean value theorem, there exists a value $x_i$ within each bin such that $f\left( {{x_i}} \right)\Delta  = \int_{i\Delta }^{\left( {i + 1} \right)\Delta } {f(x)dx} $. Then, we define a quantized random variable $X^{\Delta}$, which is given by
\begin{equation}
X^{\Delta}=x_i, \quad\quad if \,\,\,\, i\Delta \leq X < (i+1)\Delta.
\end{equation}
Therefore, $p_i=Pr\{X^{\Delta}=x_i\}=\int_{i\Delta }^{\left( {i + 1} \right)\Delta } {f(x)dx} =f(x_i)\Delta$.

The MIM of $X^{\Delta}$ is given by \cite{fan2016message}
\begin{flalign}\label{equ:Relation between RMIM and MIM}
L(X^{\Delta}) &= \log \sum\limits_{i = 1}^n {{p_i}{e^{\varpi (1 - {p_i})}}} \\
 &= \log \sum\limits_{i = 1}^n {\Delta {f(x_i)}{e^{\varpi (1 - \Delta {f(x_i)})}}} \tag{\theequation a}\label{equ:Relation between RMIM and MIM a}\\
  &= {{\mathop{\rm loge}\nolimits} ^\varpi }\sum\limits_{i = 1}^n {\Delta {f(x_i)}{e^{ - \varpi \Delta {f(x_i)}}}} \tag{\theequation b}\label{equ:Relation between RMIM and MIM b}\\
 &  = \varpi  + \log \sum\limits_{i = 1}^n {\Delta {f(x_i)}{e^{ - \varpi \Delta {f(x_i)}}}}, \tag{\theequation c}\label{equ:Relation between RMIM and MIM c}
 \end{flalign}
since $\sum_{i=1}^n {f(x_i)\Delta}=1$. Substituting $\varpi=1/\Delta$ in (\ref{equ:Relation between RMIM and MIM c}), we obtain
\begin{equation}\label{equ: Relation}
L = {1 \over \Delta } + \log \sum\limits_{i = 1}^n {{f(x_i)}{e^{ - {f(x_i)}}}} \Delta.
\end{equation}
It is observed that the first term in (\ref{equ: Relation}) approaches infinity when $\Delta \to 0$. Therefore, the MIM of continuous random variable approaches infinity, which makes no sense. However, the second term in (\ref{equ: Relation}) can help us characterize the relative importance of continuous random variables. The logarithm operator does not change the monotonicity of a function, which is only to reduce the magnitude of the numerical results, so $\sum\limits_{i = 1}^n {{f(x_i)}{e^{ - {f(x_i)}}}} \Delta $ is adopted to measure the relative importance. If $f(x)e^{-f(x)}$ is Riemann integrable, $\sum\limits_{i = 1}^n {{f(x_i)}{e^{ - {f(x_i)}}}} \Delta $ approaches the integral of $f(x)e^{-f(x)}$ as $\Delta \to 0$ by definition of Riemann integrability.

\section{The Properties of DMIM}\label{sec:three}
In this section, the properties of DMIM are discussed in details.
\subsection{Upper and Lower Bound}
For any continuous random variable $X$ with density $f(x)$, it is noted that
\begin{flalign}\label{equ:Bound}
\int_S {f(x){e^{ - f(x)}}dx}  \le \int_S {f(x)dx}  = 1.
\end{flalign}
(\ref{equ:Bound}) is obtained for the fact that $0\leq f(x)\leq 1$, which leads to $e^{-f(x)}\leq 1$. As a result, $f(x) e^{-f(x)}\leq f(x)$. 
Obviously, we also find $l(x)\geq 0$ because $f(x)\geq 0$. 
Hence, we obtain
\begin{equation}
   0\leq l(X) \leq1.
\end{equation}

\subsection{Translation with Constant}
Let $Y=X+c$, where $c$ is a real constant. Then $f_Y(y)=f_X(y-c)$, and
\begin{equation}
l(X+c)=\int_{-\infty}^{+\infty} {{f_X}(x - c){e^{ - {f_X}(x - c)}}dx} =l(X).
\end{equation}
As a result, the translation with a constant does not change the DMIM.

\subsection{Stretching}
Let $Y=aX$, where $a$ is a non-zero real number. Then ${f_Y}\left( y \right) = {1 \over {\left| a \right|}}{f_X}\left( {{y \over a}} \right)$, and
\begin{flalign}\label{equ:Stretching}
l\left( {aX} \right) &= \int {{f_Y}\left( y \right){e^{ - {f_Y}\left( y \right)}}dy} \\
 &= \int {{1 \over {\left| a \right|}}{f_X}\left( {{y \over a}} \right){e^{ - {1 \over {\left| a \right|}}{f_X}\left( {{y \over a}} \right)}}dy} \tag{\theequation a}\label{equ:Stretching a} \\
 & = \int {{f_X}\left( x \right){e^{ - {1 \over {\left| a \right|}}{f_X}\left( x \right)}}dx}. \tag{\theequation b}\label{equ:Stretching b}
\end{flalign}
Consider the extreme case, we get
\begin{equation}
\mathop {\lim }\limits_{a \to \infty } l\left( {aX} \right) = \mathop {\lim }\limits_{a \to \infty } \int {f\left( x \right){e^{ - {1 \over {\left| a \right|}}{f_X}\left( X \right)}}dx}  = 1,
\end{equation}
and
\begin{equation}
\mathop {\lim }\limits_{a \to 0} l\left( {aX} \right) = \mathop {\lim }\limits_{a \to 0} \int {f\left( x \right){e^{ - {1 \over {\left| a \right|}}{f_X}\left( X \right)}}dx}  = 0.
\end{equation}

Asymptotically, too small stretch factor will lead to lessen the relative importance of random variables. Nevertheless, when the stretch factor approaches infinity, DMIM reaches the maximum.

\subsection{Relation of DMIM to R{\'{e}}nyi Entropy}
The differential R{\'{e}}nyi entropy of a continuous random variable $X$ with density $f(x)$ is given by \cite{van2014renyi}
\begin{equation}
 h_{\alpha}(X)=\frac{1}{1-\alpha} \ln \int {(f(x))}^{\alpha}dx,
\end{equation}
where $\alpha>0$ and $\alpha \ne 1$. As $\alpha$ tends to 1, the R{\'{e}}nyi entropy tends to the Shannon entropy.

Therefore, we obtain
\begin{equation}
 \int {(f(x))}^{\alpha}dx=e^{(1-\alpha) h_{\alpha}(X)}.
\end{equation}
Hence, we find
\begin{flalign}\label{equ:Renyi}
  l(X) &= 1 + \sum\limits_{n = 1}^\infty  {{{{{\left( { - 1} \right)}^n}} \over {n!}}\int_{ - \infty }^{ + \infty } {{{\left( {f (x)} \right)}^{n + 1}}dx} } \\
  &= 1 + \sum\limits_{n = 1}^\infty  {{{{{\left( { - 1} \right)}^n}} \over {n!}}e^{-n h_{n+1}(X)} }  . \tag{\theequation d}\label{equ:Renyi a}
\end{flalign}
Obviously, the DMIM is an infinite series of R{\'{e}}nyi Entropy.

\subsection{Truncation Error}
In this part, the remainder term of (\ref{equ:thm lf}) will be discussed. In fact, the remainder term is limited in many cases, which is summarized as the following theorem.
\begin{thm}\label{thm:Truncation Error}
 If $\int {(f(x))}^{n+1}dx \leq \varepsilon $ for every $n\geq m$, then
 \begin{equation}\label{equ:Truncation Error}
 \left| {l(X) - (1 + \sum\limits_{n = 1}^{m-1}  {{{{{\left( { - 1} \right)}^n}} \over {n!}}\int_{ - \infty }^{ + \infty } {{{\left( {f (x)} \right)}^{n + 1}}dx} } )} \right| \le e\varepsilon.
 \end{equation}
\end{thm}
\begin{proof}
Substituting (\ref{equ:thm lf}) in the left of (\ref{equ:Truncation Error}), we obtian
\begin{flalign}\label{equ:Truncation Error 1}
&\left| {l\left( X \right) - \left( {1 + \sum\limits_{n = 1}^{m - 1} {{{{{\left( { - 1} \right)}^n}} \over {n!}}\int_{ - \infty }^{ + \infty } {{{\left( {f(x)} \right)}^{n + 1}}dx} } } \right)} \right| \nonumber\\
=&   \left| {\sum\limits_{n = m}^\infty  {{{{{\left( { - 1} \right)}^n}} \over {n!}}\int_{ - \infty }^{ + \infty } {{{\left( {f(x)} \right)}^{n + 1}}dx} } } \right| \\
\leq&  \sum\limits_{n = m}^\infty  {\left| {{{{{\left( { - 1} \right)}^n}} \over {n!}}\int_{ - \infty }^{ + \infty } {{{\left( {f(x)} \right)}^{n + 1}}dx} } \right|}  \tag{\theequation a}\label{equ:Truncation Error 1 a}\\
= &  \sum\limits_{n = m}^\infty  {\left| {{{1} \over {n!}}\int_{ - \infty }^{ + \infty } {{{\left( {f(x)} \right)}^{n + 1}}dx} } \right|}  \tag{\theequation b}\label{equ:Truncation Error 1 b}\\
\leq &\left( {\sum\limits_{n = m}^\infty  {{1 \over {n!}}} } \right)\varepsilon  \tag{\theequation c}\label{equ:Truncation Error 1 c}\\
\leq&\left( {1 + \sum\limits_{n = 1}^{m - 1} {{1 \over {n!}}}  + \sum\limits_{i = m}^\infty  {{1 \over {n!}}} } \right)\varepsilon \tag{\theequation d}\label{equ:Truncation Error 1 d}\\
=&e\varepsilon  ,\tag{\theequation f}\label{equ:lTruncation Error 1 e}
\end{flalign}
(\ref{equ:Truncation Error 1 c}) follows from $\int {(f(x))}^{n+1}dx \leq \varepsilon $, when $n\ge m$.
\end{proof}
That is to say, if the integral of the density to the $(n+1)$-th power is limited, the remainder term will be restricted.

\begin{cor}\label{cor:shouxian}
  If $\int {(f(x))}^{n+1}dx \leq \varepsilon $ for every $n\geq m$, then $\left| {l(X) - (1 + \sum\limits_{n = 1}^{m - 1} {\frac{{(-1)}^n}{n!}{e^{ - n{h_{n + 1}}(X)}}} )} \right| \le e\varepsilon$.
\end{cor}
\begin{proof}
 Clearly we have
 \begin{flalign}\label{equ:shouxian 1}
  &\left| {l\left( X \right) - \left( {1 + \sum\limits_{n = 1}^{m - 1} {{{{{\left( { - 1} \right)}^n}} \over {n!}}{e^{ - n{h_{n + 1}}(X)}}} } \right)} \right| \nonumber  \\
    =&\left| {l(X) - (1 + \sum\limits_{n = 1}^{m-1}  {{{{{\left( { - 1} \right)}^n}} \over {n!}}\int_{ - \infty }^{ + \infty } {{{\left( {f (x)} \right)}^{n + 1}}dx} } )} \right| \\
   \leq& e\varepsilon ,   \tag{\theequation a}\label{equ:shouxian 1 a}
\end{flalign}
(\ref{equ:shouxian 1 a}) follows from Theorem \ref{thm:Truncation Error}.
 \end{proof}

 \begin{rem}\label{rem:jieduan}
 Letting $m=2$ in (\ref{equ:Truncation Error 1 c}), after manipulations, we obtain
 \begin{flalign}
\left| l(X)-(1-e^{-h_2(X)})\right| &\le (e-2) \varepsilon,\\
  l(X) + e^{-h_2(X)} & \le 1+(e-2)\varepsilon. \tag{\theequation a}
 \end{flalign}
 \end{rem}
Especially, if $\varepsilon$ is too small, we have $ l(x) + e^{-h_2(X)}  \approx 1$. That means $l(X)$ is approximately the dual part of R{\'{e}}nyi entropy with order $2$.

\section{The DMIM of Some Distributions}  \label{sec:four}
\subsection{Uniform Distribution}
For a random variable whose density is $\frac{1}{b-a}$ for $a\le x\le b$ and $0$ elsewhere, we have
\begin{equation}\label{equ:Uniform 1}
l(X) =  {\int_a^b {{1 \over {b - a}}{e^{ - {1 \over {b - a}}}}dx} } = {e^{ - {1 \over {b - a}}}}.
\end{equation}
Note that
\begin{flalign}\label{equ:Uniform 2}
\mathop {\lim }\limits_{\left( {b - a} \right) \to 0} {e^{ - {1 \over {b - a}}}} =0,\\
\mathop {\lim }\limits_{\left( {b - a} \right) \to \infty } {e^{ - {1 \over {b - a}}}} = 1.\tag{\theequation a}\label{equ:Uniform a}
\end{flalign}

\subsection{Normal Distribution}
Let $X\sim \phi (x) = {1 \over {\sqrt {2\pi {\sigma ^2}} }}{e^{ - {{{{\left( {x - \mu } \right)}^2}} \over {2{\sigma ^2}}}}}$ with $\sigma \ne 0$, then
\begin{flalign}\label{equ:Normal 1}
\int_{ - \infty }^{ + \infty } &{{{\left( {\phi (x)} \right)}^{n + 1}}dx}  = \int_{ - \infty }^{ + \infty } {{{\left( {{1 \over {\sqrt {2\pi {\sigma ^2}} }}{e^{ - {{{{\left( {x - \mu } \right)}^2}} \over {2{\sigma ^2}}}}}} \right)}^{n + 1}}dx} \\
 &={\left( {{1 \over {\sqrt {2\pi {\sigma ^2}} }}} \right)^{n + 1}}\int_{ - \infty }^{ + \infty } {{e^{ - {{n + 1} \over {2{\sigma ^2}}}{{\left( {x - \mu } \right)}^2}}}dx}  \tag{\theequation a}\label{equ:Normal 1 a}\\
 & = {\left( {{1 \over {\sqrt {2\pi {\sigma ^2}} }}} \right)^{n + 1}}\sqrt {{{2\pi {\sigma ^2}} \over {n + 1}}}   \tag{\theequation b}\label{equ:Normal 1 b}\\
 &  = {1 \over {\sqrt {n + 1} }}{\left( {{1 \over {\sqrt {2\pi {\sigma ^2}} }}} \right)^n}.  \tag{\theequation c}\label{equ:Normal 1 c}
   \end{flalign}
Substituting (\ref{equ:Normal 1 c}) in (\ref{equ:thm lf}), we obtian
\begin{equation}\label{equ:normal RMIM}
l(X) = 1 + \sum\limits_{n = 1}^\infty  {{{{{\left( { - 1} \right)}^n}} \over {n!}}} {1 \over {\sqrt {n + 1} }}{\left( {{1 \over {\sqrt {2\pi {\sigma ^2}} }}} \right)^n}.
\end{equation}

\subsubsection{When $\sigma$ is large}
If $\sigma>1/\sqrt{2 \pi}$, $\int_{ - \infty }^{ + \infty } {{{\left( {\phi (x)} \right)}^{n + 1}}dx}$ will be less than or equal to $1/{(2\sqrt{3}\pi\sigma^2})$ for every $n\geq 2$ because $ {1 \over {\sqrt {n + 1} }}{\left( {{1 \over {\sqrt {2\pi {\sigma ^2}} }}} \right)^n}$ monotonically decreases in this case. According to Remark \ref{rem:jieduan}, we obtain
\begin{equation}\label{equ:normal jieduan 1}
  \left| {l(X) - (1 -e^{-h_2(X)} )} \right| \leq \frac{(e-2) }{{2\sqrt{3}\pi\sigma^2}}.
\end{equation}
If $\sigma$ is big enough, $\frac{(e-2) }{{2\sqrt{3}\pi\sigma^2}} \approx 0$. Moreover, the intensity of approximation error decreases as the inverse square of $\sigma$. In this case, substituting $h_2(X)=\ln2+0.5\ln\pi +\ln \sigma$ in $1 -e^{-h_2(X)}$, we find
\begin{equation}\label{equ:normal jieduan 2}
1 -e^{-h_2(X)} = 1-\frac{1}{2\sqrt{\pi}\sigma } \approx e^{-\frac{1}{2\sqrt{\pi}\sigma}}.
\end{equation}
We define
\begin{flalign}\label{equ: Normal big approx 1}
  {\tilde l}_1(X)&=1-\frac{1}{2\sqrt{\pi}\sigma } ,  \\
  {\tilde l}_2(X)&=e^{-\frac{1}{2\sqrt{\pi}\sigma}} .   \tag{\theequation a} \label{equ: Normal big approx 1 a}
\end{flalign}
According to (\ref{equ:normal jieduan 1}), ${\tilde l}_1(x)$ and ${\tilde l}_2(x)$ is very good approximate values for DMIM of normal distribution when $\sigma$ is not too small, which will be shown by the numerical results in section \ref{sec:numerical}.

\subsubsection{When $\sigma$ is small}
However, the DMIM of normal distribution will be hard to calculate when $\sigma$ is small. By Stirling formula, $l(X)$ can also be written as
\begin{equation}
l\left( X\right) \approx 1 + {\sum\limits_{n = 1}^\infty  {{1 \over {\sqrt {2\pi n\left( {n + 1} \right)} }}\left( { - {e \over {\sqrt {2\pi } \sigma n}}} \right)} ^n}.
\end{equation}
If $n > {e \over {\sqrt {2\pi } \sigma }}\approx{1.0844\over{\sigma}}$, we will obtain $\left| { - {e \over {\sqrt {2\pi } \sigma n}}} \right| < 1$. Let $n_0=\left\lfloor \frac{e}{2\pi \sigma} \right\rfloor$ where $\left\lfloor x \right\rfloor$ is the largest integer smaller than or equal to $x$. In this case, we define
\begin{equation}\label{equ:DMIM_approx-hat}
\hat{l}(X) = 1 + \sum\limits_{n = 1}^{n_0}  {{{{{\left( { - 1} \right)}^n}} \over {n!}}} {1 \over {\sqrt {n + 1} }}{\left( {{1 \over {\sqrt {2\pi {\sigma ^2}} }}} \right)^n},
\end{equation}
as the approximate value when $\sigma$ is small. The following theorem shows the validity of $\hat{l}(X)$.

\begin{thm}\label{thm:normal small sigma}
$X$ is a normal random variable with mean $\mu$ and variance $\sigma^2$. $\hat{l}(X)$ is the first $n_0$ terms of $l(X)$, given by (\ref{equ:DMIM_approx-hat}), where $n_0=\left\lfloor \frac{e}{2\pi \sigma} \right\rfloor$. If $\sigma$ is relatively small, Then we have
\begin{equation}
\left| {l(X) - \hat{l}(X)} \right| < {{3\sigma } \over e}.
\end{equation}
\end{thm}
\begin{proof}
Refer to the Appendix \ref{Appendices A}.
\end{proof}

It is easy to see that the upper bound of error approaches $0$ if $\sigma$ approaches $0$.

\subsection{Exponential Distribution}
Letting
\begin{equation}
\begin{split}
 X\sim f(x)=\left\{
   \begin{aligned}
 & \lambda e^{-\lambda x}, \quad x \geq 0 \\
 & 0,\quad\quad\quad x<0 \\
   \end{aligned}
   \right.,
   \end{split}
\end{equation}
where $\lambda>0$, we obtain
\begin{flalign}\label{equ:Negative Exponential 1}
\int_{ - \infty }^{ + \infty } {{{\left( {f(x)} \right)}^{n + 1}}dx}  &= \int_0^{ + \infty } {{{\left( {\lambda {e^{ - \lambda x}}} \right)}^{n + 1}}dx}  \\
 &= \int_0^{ + \infty } {{\lambda ^{n + 1}}{e^{ - \lambda \left( {n + 1} \right)x}}dx} \tag{\theequation a}\label{equ:Negative Exponential 1 a}\\
&= {\lambda ^{n + 1}}\int_0^{ + \infty } {{e^{ - \lambda \left( {n + 1} \right)x}}dx} \tag{\theequation b}\label{equ:Negative Exponential 1 b}\\
&   = {{{\lambda ^n}} \over {n + 1}}  .\tag{\theequation c}\label{equ:Negative Exponential 1 c}
 \end{flalign}
Substituting (\ref{equ:Negative Exponential 1 c}) in (\ref{equ:thm lf}), we obtain
\begin{equation}\label{equ: Negative Exponential RMIM}
l(X) = 1 + \sum\limits_{n = 1}^\infty  {{{\left( { - 1} \right)}^n}{{{\lambda ^n}} \over {\left( {n + 1} \right)!}}} =\frac{1}{\lambda}(1-e^{-\lambda}).
\end{equation}
It is noted that
\begin{flalign}\label{equ:Negative Exponential 2}
 \mathop {\lim }\limits_{\lambda  \to 0} {1 \over \lambda }\left( {1 - {e^{ - \lambda }}} \right) = 1,\\
 \mathop {\lim }\limits_{\lambda  \to \infty} {1 \over \lambda }\left( {1 - {e^{ - \lambda }}} \right) = 0. \tag{\theequation a}\label{equ:Negative Exponential 2 a}
 \end{flalign}

\subsection{Gamma Distribution}
In many cases, the $\Gamma$ distribution can be used to describe the distribution of the amount of time one has to wait until a total of $n$ events has occurred in practice \cite{ross2014first}. For a random variable obeying $\Gamma$ distribution, its density is
\begin{equation}
\begin{split}
 X\sim f(x)=\left\{
   \begin{aligned}
 & {{\lambda {e^{ - \lambda x}}{{\left( {\lambda x} \right)}^{\alpha  - 1}}} \over {\Gamma \left( \alpha  \right)}}, \quad x \geq 0 \\
 & 0,\quad\quad\quad x<0 \\
   \end{aligned}
   \right.,
   \end{split}
\end{equation}
where $\lambda,\alpha>0$, we obtain
\begin{flalign}\label{equ:gamma 1}
  &\int_{-\infty}^{ + \infty } {{{\left( {f(x)} \right)}^{n + 1}}dx}  \nonumber \\
  =& \int_0^{ + \infty } {{{\left( {{{\lambda {e^{ - \lambda x}}{{\left( {\lambda x} \right)}^{\alpha  - 1}}} \over {\Gamma \left( \alpha  \right)}}} \right)}^{n + 1}}dx}  \\
  = &  {{{\lambda ^n}} \over {{{\left( {n + 1} \right)}^{\alpha n - n + \alpha }}{\Gamma ^{n + 1}}\left( \alpha  \right)}}\int_0^{ + \infty } {{e^{ - t}}{t^{\left( {\alpha  - 1} \right)\left( {n + 1} \right)}}dt}  \tag{\theequation a}\label{equ:gamma 1 a} \\
   = & {{{\lambda ^n}\Gamma \left( {\alpha n - n + \alpha } \right)} \over {{{\left( {n + 1} \right)}^{\alpha n - n + \alpha }}{\Gamma ^{n + 1}}\left( \alpha  \right)}} \tag{\theequation b}.\label{equ:gamma 1 b}
\end{flalign}
Substituting (\ref{equ:gamma 1 b}) in (\ref{equ:thm lf}), we obtian
\begin{equation}\label{equ:gamma 2}
l(X) = 1 + \sum\limits_{n = 1}^\infty  {{{{{\left( { - 1} \right)}^n}} \over {n!}}{{{\lambda ^n}\Gamma \left( {\alpha n - n + \alpha } \right)} \over {{{\left( {n + 1} \right)}^{\alpha n - n + \alpha }}{\Gamma ^{n + 1}}\left( \alpha  \right)}}}.
\end{equation}

\subsection{Beta Distribution}
The $\beta$ distribution often arises to depict a random variable whose set of possible values is some finite interval, such as $[0,1]$ \cite{ross2014first}. For a random variable follows $\beta$ distribution whose density is
\begin{equation}
\begin{split}
 X\sim f(x)=\left\{
   \begin{aligned}
 & {1 \over {B(a,b)}}{x^{a - 1}}{\left( {1 - x} \right)^{b - 1}}, \quad  0<x<1 \\
 & 0,\quad\quad\quad else \\
   \end{aligned}
   \right.,
   \end{split}
\end{equation}
where $B(a,b) = \int_0^1 {{x^{a - 1}}{{\left( {1 - x} \right)}^{b - 1}}dx} $ and $a,b>0$. According to \cite{ross2014first}, we have
\begin{equation}
B(a,b) = {{\Gamma \left( a \right)\Gamma \left( b \right)} \over {\Gamma \left( {a + b} \right)}}.
\end{equation}
In fact, we find
\begin{flalign}\label{equ:beta 1}
&\int_{ - \infty }^{ + \infty } {{{\left( {f(x)} \right)}^{n + 1}}dx} \nonumber\\
& = \int_0^1 {{{\left( {{1 \over {B(a,b)}}{x^{a - 1}}{{\left( {1 - x} \right)}^{b - 1}}} \right)}^{n + 1}}dx} \\
 &= {1 \over {{B^{n + 1}}(a,b)}}\int_0^1 {{x^{\left( {a - 1} \right)\left( {n + 1} \right)}}{{\left( {1 - x} \right)}^{\left( {b - 1} \right)\left( {n + 1} \right)}}dx}  \tag{\theequation a}\label{equ:beta 1 a}\\
 & = {{B(an - n + a,bn - n + b)} \over {{B^{n + 1}}(a,b)}} \nonumber \\
 &\quad \quad \quad \quad \quad \cdot \int_0^1 {{{{x^{\left( {a - 1} \right)\left( {n + 1} \right)}}{{\left( {1 - x} \right)}^{\left( {b - 1} \right)\left( {n + 1} \right)}}} \over {B(an - n + a,bn - n + b)}}dx}   \tag{\theequation b} \\
 & = {{B(an - n + a,bn - n + b)} \over {{B^{n + 1}}(a,b)}},\tag{\theequation c}\label{equ:beta 1 b}
\end{flalign}
Hence, we obtain
\begin{equation}\label{equ:beta 2}
 l(X) = 1 + \sum\limits_{n = 1}^\infty  {{{{{\left( { - 1} \right)}^n}} \over {n!}}{{B(an - n + a,bn - n + b)} \over {{B^{n + 1}}(a,b)}}}.
\end{equation}

\subsection{Laplace Distribution}
A random variable, whose density function is
\begin{equation}
  f(x)=\frac{\lambda}{2} e^{\lambda \left| x-\theta \right|},
\end{equation}
has a Laplace distribution where $\theta$ is a location parameter and $\lambda>0$.
In fact, we find
\begin{flalign}\label{equ:laplace 1}
{\int_{ - \infty }^{ + \infty } {{{\left( {{\lambda  \over 2}{e^{ - \lambda \left| {x - \theta } \right|}}} \right)}^{n + 1}}dx = {1 \over {n + 1}}\left( {{\lambda  \over 2}} \right)} ^n}.
\end{flalign}
Substituting (\ref{equ:laplace 1}) in (\ref{equ:thm lf}), we obtian
\begin{equation}\label{equ:laplace 2}
l(X) = 1 + \sum\limits_{n = 1}^\infty  {{1 \over {\left( {n + 1} \right)!}}} {\left( { - {\lambda  \over 2}} \right)^n} = {2 \over \lambda }\left( {1 - {e^{ - {\lambda  \over 2}}}} \right).
\end{equation}

For simplicity to follow, the DMIM for these common densities are summarized in Table \ref{tab:result1}.

\begin{table*}[tbp]
\centering
    \caption{Table of DMIM for common densities.}\label{tab:result1}
\begin{tabular}{c|c|c|c}
\toprule [1 pt]
Distribution & Parameter & Density & DMIM \\
\hline
Uniform &  $a$,$b$ & $f(x)=\frac{1}{b-a}, \, a\le x\le b$ & $e^{-\frac{1}{b-a}}$\\
\hline
Normal  &  $\mu$,$\sigma$ & $f(x)={1 \over {\sqrt {2\pi {\sigma ^2}} }}{e^{ - {{{{\left( {x - \mu } \right)}^2}} \over {2{\sigma ^2}}}}}$  & $1 + \sum\limits_{n = 1}^\infty  {{{{{\left( { - 1} \right)}^n}} \over {n!}}} {1 \over {\sqrt {n + 1} }}{\left( {{1 \over {\sqrt {2\pi {\sigma ^2}} }}} \right)^n}$\\
\hline
Exponential &$\lambda$ & $f(x)=\lambda e^{-\lambda x}, \quad x,\lambda> 0$  &$\frac{1}{\lambda}(1-e^{-\lambda})$\\
\hline
\multirow{2}{*}{Gamma}  & \multirow{2}{*}{$\alpha $,$\lambda$} & $f(x)={{\lambda {e^{ - \lambda x}}{{\left( {\lambda x} \right)}^{\alpha  - 1}}} \over {\Gamma \left( \alpha  \right)}} $  &\multirow{2}{*}{$1 + \sum\limits_{n = 1}^\infty  {{{{{\left( { - 1} \right)}^n}} \over {n!}}{{{\lambda ^n}\Gamma \left( {\alpha n - n + \alpha } \right)} \over {{{\left( {n + 1} \right)}^{\alpha n - n + \alpha }}{\Gamma ^{n + 1}}\left( \alpha  \right)}}}$} \\
& & $ x \geq 0, \lambda,\alpha>0$ & \\
\hline
\multirow{2}{*}{Beta} &\multirow{2}{*}{$a$,$b$}& $f(x)={1 \over {B(a,b)}}{x^{a - 1}}{\left( {1 - x} \right)^{b - 1}} $ & \multirow{2}{*}{$1 + \sum\limits_{n = 1}^\infty  {{{{{\left( { - 1} \right)}^n}} \over {n!}}{{B(an - n + a,bn - n + b)} \over {{B^{n + 1}}(a,b)}}} $}\\
& & $\quad  0<x<1, a,b>0$ & \\
\hline
\multirow{2}{*}{Laplace} & \multirow{2}{*}{$\lambda$,$\theta$}  & $f(x)=\frac{\lambda}{2} e^{\lambda \left| x-\theta \right|}$ &  \multirow{2}{*}{$\frac{2}{\lambda} \left( {1-e^{-\frac{\lambda}{2}}}\right)$} \\
& & $-\infty<x,\theta<\infty, \lambda>0$& \\


\toprule [1 pt]
\end{tabular}
\end{table*}

\section{Goodness of Fit with DMIM}\label{sec:five}
In this section, we will consider the problem of distribution goodness-of-fit in a continuous setting. Let $X_1,X_2,...X_n$ be a sequence of independent and identically distributed random variables, each having mean $\mu$ and variance $\sigma^2$. In practice, the real distribution is generally unknown and we usually use empirical distribution to substitute real distribution. Generally, the empirical distribution function is given by
\begin{equation}
 \hat F_n(x)=\frac{1}{n}\sum\limits_{k=1}^n {I_{(X_k\leq x)}},
\end{equation}
and the real distribution is $F(x)$ .

One practical problem that can occur with this strategy is that how many samples is required for fitting the real distribution with an acceptable bias in some degree. Many literatures studied this problem by Kolmogorov-Smirnov statistic \cite{massey1951kolmogorov,lilliefors1967kolmogorov,Resnick1992Advantures}. When $n$ is big enough, the confidence limits for a cumulative distribution are given by \cite{Resnick1992Advantures},
\begin{equation}
 P\{D_n>d\}\approx 2\sum\limits_{k=1}^{\infty} {{(-1)}^{k-1}e^{-2nk^2d^2}},
\end{equation}
where $D_n$ is error bound between empirical distribution and real distribution, called Kolmogorov-Smirnov statistic, which is defined as
\begin{equation}
D_n=\mathop {\sup}\limits_x {\left| \hat F_n (x)-F(x)\right|},
\end{equation}

Though this result can describe the goodness of fit very well and guide us to choose the sampling numbers, we need to give two artificial criterions, the deviation value $d$ and the probability $ P\{D_n>d\}$, in order to determine $n$. In addition, this method do not take the message importance of samples into account, which makes the process of information collection not intuitionistic.

In this paper, we consider this problem from the perspective of DMIM. Firstly, we define
\begin{equation}\label{equ:XNX define}
\gamma \left( n \right)  ={{l\left( {\sum\limits_{i = 1}^n {{X_i}} } \right)} / {l(X)}}  .
\end{equation}
as relative importance of these $n$ sample points. According to central-limit theorem \cite{ross2014first}, when $n$ is big enough, $\sum\nolimits_{i =1}^n {{X_i}} $ approximately obeys normal distribution $N(n\mu,n\sigma^2)$. In fact, when $\sqrt{n}\sigma$ is not too small (such a condition is satisfied because $n$ is big enough), $l\left( {\sum\nolimits_{i = 1}^n {{X_i}} } \right) \approx {e^{ - {1 \over {2\sqrt {\pi n} \sigma }}}}$ according to (\ref{equ:normal jieduan 2}). Hence
\begin{equation}\label{equ:XNX define 1}
\gamma(n)={{{e^{ - {1 \over {2\sqrt {\pi n} \sigma }}}}} \over {l\left( X \right)}}.
\end{equation}

We find $\gamma(n)$ increases rapidly firstly, and then increases slowly by analyzing its monotonicity. Moreover, we obtain
\begin{equation}\label{equ:XNX infty}
\gamma(\infty)=\mathop {\lim }\limits_{n \to \infty } \gamma \left( n \right) = \mathop {\lim }\limits_{n \to \infty } {{{e^{ - {1 \over {2\sqrt {\pi n} \sigma }}}}} \over {l\left( X \right)}} = {1 \over {l\left( X \right)}},
\end{equation}
which means $\gamma(n)$ reaches limit as $n \to \infty$. In fact, these two points are consistent with the characteristic of data fitting. Both $\gamma(n)$ and data fitting have the law of diminishing of marginal utility. Furthermore, the goodness of fit can not increase unboundedly and it reaches the upper bound when the number of sampling points approaches infinity. DMIM is bounded, while Shannon entropy and R{\'e}nyi entropy do not possess these characteristic. In conclusion, we adopt $\left| \gamma(\infty)-\gamma(n)\right|$ to describe the goodness of fit.


\begin{thm}\label{thm:low bound number}
  $X_1,X_2,X_3,\dots,X_n$ are the $n$ sampling of a continuous random variable $X$, whose density is $f(x)$. If $\left| \gamma(\infty)-\gamma(n)\right| \leq \varepsilon$, we will obtain
  \begin{equation}
   P\left\{D_n>\sqrt{2\pi\sigma^2 \ln{\frac{19}{9\beta}}} \ln{\frac{1}{1-\varepsilon}}\right\}\leq \beta.
  \end{equation}
\end{thm}
\begin{proof}
Refer to the Appendix \ref{Appendices B}.
\end{proof}

\begin{rem}
According to (\ref{equ:NMIM_proof d}) in Appendix \ref{Appendices B}, we obtain
 \begin{flalign}\label{equ:NMIM_proof relation}
   \varepsilon  &= 1 - {e^{-d{{\left( {2\pi {\sigma ^2}\ln {{19} \over {9\beta }}} \right)}^{ - 1/2}}}},   \\
   \beta  &= {{19} \over 9}{e^{ - {{{d^2}} \over {2\pi {\sigma ^2}{{\ln }^2}(1 - \varepsilon )}}}}   \tag{\theequation a}\label{equ:NMIM_proof relation a}.
\end{flalign}
Therefore, there is a ternary relation among $d$, $\beta$ and $\varepsilon$. If two of them are known, the third one can be obtained, easily.
\end{rem}

\begin{rem}
For arbitrary positive number $d$ and $\beta \le1$, one can always find a $\varepsilon_0$, which can be obtained by (\ref{equ:NMIM_proof relation}), when $\varepsilon \le \varepsilon_0$, $P\left\{ {{D_n} > d} \right\} < \beta $ holds.
\end{rem}

\begin{rem}
When $\varepsilon$ tends zero, which means $n\to \infty$, at this time, $P\left\{D_n>0\right\}=0$. Therefore, the real distribution is equal to empirical distribution with probability $1$ as $\varepsilon \to 0$. That is,
\begin{equation}
\hat F_n(x) \to F(x) \quad as \quad \varepsilon \to 0.
\end{equation}
\end{rem}

Actually, the DMIM deviation characterizes the process of collection information in terms of data structure. With the growth of sampling number, the information gathers, and the empirical distribution approaches real distribution at the same time. In particular, when $n \to \infty$, all the information about the real distribution will be obtained. In this case, the empirical distribution is equal to real distribution, naturely.

\begin{rem}
For arbitrary continuous random variable with variance $\sigma^2$, if the maximal allowed DMIM deviation is $\varepsilon$, the sampling number should be bigger than $1/(4 \pi \sigma^2 \ln^2(1-\varepsilon) $ according to (\ref{equ:number_choose}).
\end{rem}
The sampling number only depends on one artificial criterion, the DMIM deviation, while the variance are the own attributes of the observed variable $X$. Furthermore, the sampling number in the new developed method has nothing to do with the distribution form, which means the new method is distribution-free.

\section{Numerical Results}\label{sec:numerical}
In this section, we present some numerical results to validate the above results in this paper.

\subsection{The DMIM of Normal Distribution}
 First of all, we analyze the DMIM in normal distribution by simulation. Its standard deviation $\sigma$ is varing from $0.01$ to $10$.

 Fig. \ref{fig:normal_big_sigma_2} depicts the DMIM versus standard deviation $\sigma$ in Normal distribution. We observe that there are some constraints on DMIM in this case. That is, the DMIM grows with the increasing of $\sigma$. Furthermore, it increases rapidly when $\sigma$ is small ($\sigma<1$), while it increases slowly when $\sigma$ is big ($\sigma>4$). Besides, DMIM is non-negative and it is very close to zero when $\sigma$ approaches zero. In order to avoid complex calculations, we give two approximate value of DMIM in Gauss distribution, which are $\tilde{l}_1(X)=1-1/(2\sqrt{\pi}\sigma)$ and $\tilde{l}_2(X)=e^{-1/(2\sqrt{\pi}\sigma)}$. Obviously, the gap between true value $l(X)$ and the approximate value $\tilde{l}_1(X)$ will be very small if $\sigma$ is big enough. However, $\tilde{l}_1(X)$ is smaller than $l(X)$ when $\sigma$ is small. For $\tilde{l}_2(x)$, the gap between it and $l(X)$ will be very small if $\sigma$ is not too small. In fact, there is only a slight deviation between them when $\sigma$ is small.

\begin{figure}
  \centerline{\includegraphics[width=8.0cm]{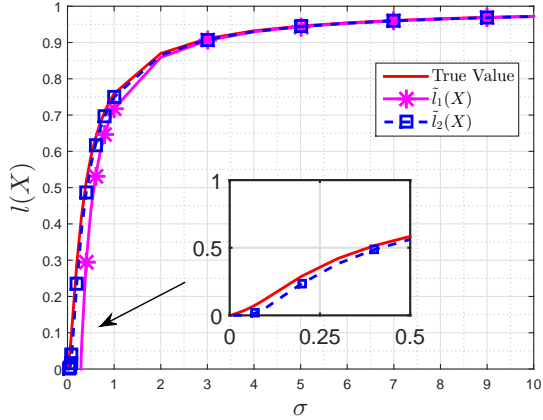}}
  \caption{$l(X)$ vs. $\sigma$ in normal distribution.}\label{fig:normal_big_sigma_2}
\end{figure}

 Moreover, Fig. \ref{fig:normal_big_sigma_err2} shows the absolute and relative error when we adopt approximations. Some observations are obtained. Both absolute and relative error are relatively small when $\sigma$ is big ($\left|l(X)-\tilde{l}_1(X)\right|/l(X)<1\%$ when $\sigma>2.5$, and $\left|l(X)-\tilde{l}_2(X)\right|/l(X)<1\%$ when $\sigma>1.25$), and they both decreases with increasing of $\sigma$ for two approximate values in most of time. In fact, when $\sigma$ is not too small ($\sigma>0.3$), the relative error of $\tilde{l}_2(X$ is smaller than $10\%$. When $\sigma<6.25$, the relative error of $\tilde{l}_2(X)$ is smaller than that of $\tilde{l}_2(X)$ and the opposite is true when $\sigma>6.25$. In summary, $\tilde{l}_2(X)$ is a good approximation for all the $\sigma$ and $\tilde{l}_1(X)$ is an excellent approximation when $\sigma$ is big enough.

\begin{figure}
  \centerline{\includegraphics[width=9.5cm]{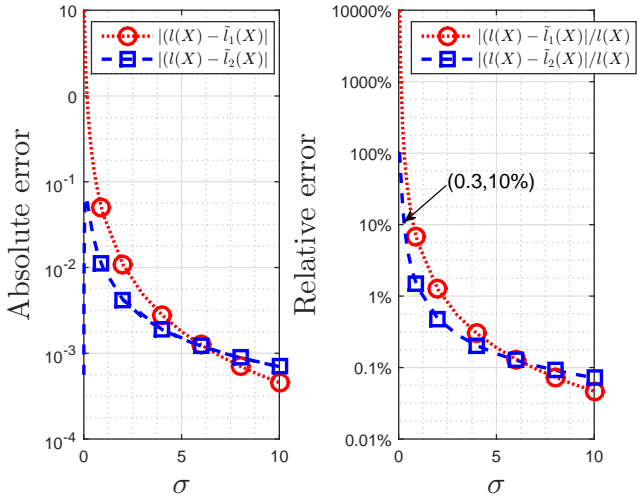}}
  \caption{Error vs. $\sigma$ in normal distribution.}\label{fig:normal_big_sigma_err2}
\end{figure}

Fig. \ref{fig:normal_small_sigma_1} shows the truncation error $|l(X)-\hat{l}(X)|$ versus the number of series $N$ when $\sigma$ is small. Without loss of generality, we take $\sigma$ as $0.02$, $0.03$ and $0.05$. $N$ is varing from $0$ to $100$. Some interesting observations are made. The truncation error will remain unchanged and approach zero only if $N>N_0$, such as $n>70$ when $\sigma=0.02$. When $N<N_0$, it increases at first and then decreases. It also can be seen that $N_0$ decreases with the increasing of $\sigma$. Furthermore, for the same $N$, the DMIM decreases with the increasing of $\sigma$. Furthermore, $\hat{l}(x)$, which is given by (\ref{equ:DMIM_approx-hat}), is a good approximation because $|l(X)-\hat{l}(X)|<0.01<3\sigma/e$ when $N=n_0$.

\begin{figure}
  \centerline{\includegraphics[width=8.0cm]{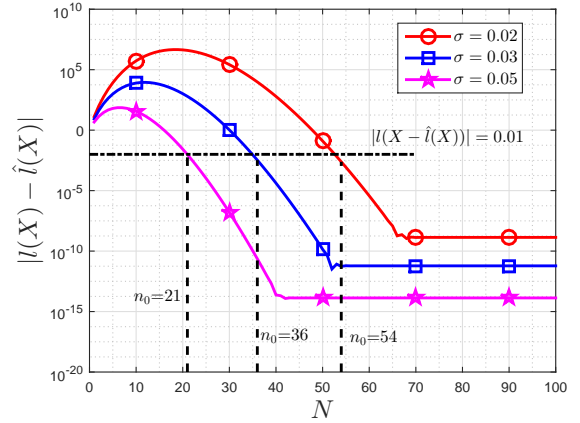}}
  \caption{$|l(x)-\hat{l}(x)|$ vs. $N$.}\label{fig:normal_small_sigma_1}
\end{figure}

\subsection{The DMIM for Common Densities}
Fig. \ref{fig:smae_var} shows the DMIM of uniform distribution, normal distribution, exponential distribution, Gamma distribution and Laplace distribution when the variance increases from $0.1$ to $100$. The simulation parameter in $\Gamma$ distribution is set as $\alpha=0.5,1.5$. It is observed that the DMIM increases with the increasing of variance for all these distributions. Among them, the DMIM of normal distribution is the largest and that of Gamma distribution ($\alpha=0.5$) is the smallest. Fig. \ref{fig:smae_var} also shows that the DMIM of Gamma distribution increases with increasing of $\alpha$ for the same variance. It also can be seen from the figure, that the gap between the DMIM of uniform distribution and that of normal distribution is negligibly small when variance is big enough. This is because that, for the same variance $\sigma^2$, these two DMIM respectively are $e^{-1/(2\sqrt{3}\sigma)}$ and $e^{-1/(2\sqrt{\pi}\sigma)}$ (approximate value when $\sigma$ is large according to (\ref{equ:normal jieduan 2})), which are very close.

\begin{figure}
  \centerline{\includegraphics[width=8.0cm]{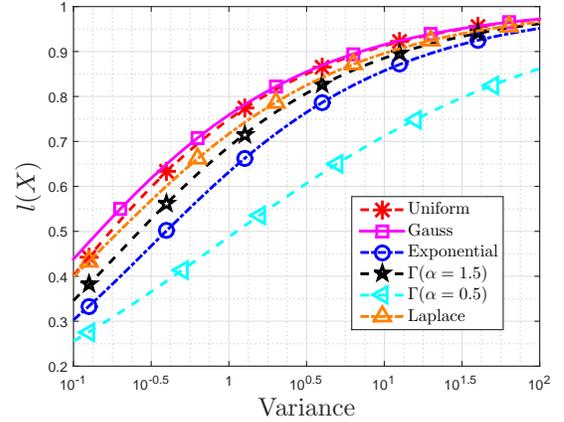}}
  \caption{$l(X)$ vs. Variance.}\label{fig:smae_var}
\end{figure}

\subsection{Goodness-of-fit with DMIM}
Next we focus on conducting Monte Carlo simulation by computer to validate our results about goodness of fit. The samples are drawn by independent identically distributed Gaussian, each having mean zero. Their standard deviation is $1$ or $2$. The DMIM deviation $\varepsilon$ is varying from $0.001$ to $0.1$. The confidence limit $\beta$ is $0.001$. For each value of $\varepsilon$, the simulation is repeated $10000$ times.

Fig. \ref{fig:Number_normal} shows the relationship between the probability of error bound $P\{D>d\}$ and DMIM deviation $\varepsilon$. Some observations can be obtained. The probability of error bound decreases with the decreasing of DMIM deviation. In fact, this process can be divided into three phases. In phase one, in which $\varepsilon$ is very small ($\varepsilon<10^{-2.8}$ when $d=0.01$ and $\sigma=1$), $P\{D>d\}$ is close to zero. In phase two, $\varepsilon$ is neither too small nor too large ($10^{-2.8}<\varepsilon<10^{-2}$ when $d=0.01$ and $\sigma=1$). In this case, $P\{D>d\}$ increases rapidly from zero to one. In the phase three, in which $\varepsilon$ is large ($\varepsilon>10^{-2}$ when $d=0.01$ and $\sigma=1$), $P\{D>d\}$ approaches one. For the same standard deviation, $P\{D>d\}$ decreases with increasing of $d$ when $P\{D>d\}<1$. Furthermore, for the same $d$, the probability of error bound increases with increasing of the standard deviation.

\begin{figure}
  \centerline{\includegraphics[width=8.0cm]{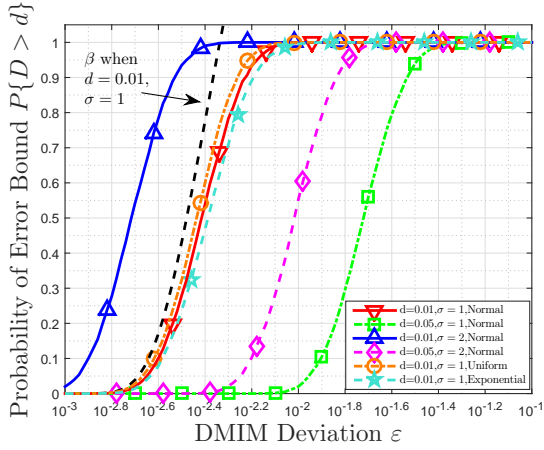}}
  \caption{Probability of error bound $P\{D>d\}$ vs. DMIM deviation $\varepsilon$.}\label{fig:Number_normal}
\end{figure}

The simulation results of $P\{D>0.01\}$ are listed in Table. \ref{tab:result2}, where the sampling number $n$ is given by (\ref{equ:number_choose}) and the upper bound for the error probability $\beta$ is given by (\ref{equ:NMIM_proof relation a}). In this table, we take $d=0.01$ as the criterion to evaluate the error between the empirical distribution and real distribution. To better validate our results, normal distribution, exponential distribution, uniform distribution and Laplace distribution are listed here. The standard deviation of these four distribution is $\sigma$. As a result, $\lambda=1/\sigma$ in exponential distribution, and the density of uniform distribution is $1/(2\sqrt{3}\sigma)$. The $\lambda=\sqrt{2}/\sigma$ in Laplace distribution. The remaining parameter values are same with that in Fig. \ref{fig:Number_normal}. We obtain that $\beta$ is indeed the upper bound of $P\{D>0.01\}$ because every $P\{D>0.01\}$ is smaller than $\beta$. For each distribution, it is noted that the sampling number increases with the decreasing of DMIM deviation. In addition, $P\{D>0.01\}$ decreases with decreasing of the DMIM deviation. $P\{D>0.01\}$ can even be zero when $\varepsilon=0.001$ and $\sigma=1$. For the same DMIM deviation, $\beta$ and $P\{D>0.01\}$ increase with increasing of $\sigma$. Therefore, if one wants to have the same precision in different variance, it needs to select smaller $\varepsilon$ when $\sigma$ is larger, such as $\varepsilon=0.002$ when $\sigma=1$ and $\varepsilon=0.001$ when $\sigma=2$. Furthermore, when $n$ is not too small, for the same $\varepsilon$ and $\sigma$, $P\{D>0.01\}$ of these four distribution is very close to each other, which means this method is distribution-free.




\begin{table*}[tbp]
\centering
    \caption{Table of probability of error bound $P\{D>0.01\}$. The sampling number $n$ is given by (\ref{equ:number_choose}) and the upper bound for the error probability $\beta$ is given by (\ref{equ:NMIM_proof relation a}).}\label{tab:result2}
\begin{tabular}{cc|c|ccc|ccc|c}
\toprule [1 pt]
&Distribution & DMIM deviation $\varepsilon$ & \multicolumn{3}{|c|}{$\sigma=1$} & \multicolumn{3}{|c|}{$\sigma=2$}& \\
\cline{4-9}
& & & $n$ & $\beta$ &  $P\{D>0.01\}$ &$n$ & $\beta$ &  $P\{D>0.01\}$&\\
\hline
&\multirow{4}{*}{Normal} & 0.01& 787 &1.8034 &0.9994&196 &2.0296 &1&\\
&& 0.003 & 8815 &0.3621 &0.2286 &2203 &01.3586 &0.9056&\\
&& 0.002 & 19854 &0.0398 &0.0207 &4963 &0.7823 &0.5390&\\
&& 0.001 & 79497 &2.63e-7 &0 &19874 &0.0396 &0.0221&\\
\hline
&\multirow{4}{*}{Exponent} & 0.01 & 787 &1.8034 &0.9964&196 &2.0296 &1&\\
&& 0.003 & 8815 &0.3621 &0.2011 &2203 &1.3586 &0.8609&\\
&& 0.002 & 19854 &0.0398 &0.0164 &4963 &0.7823 &0.4821&\\
&& 0.001 & 79497 &2.63e-7 &0 &19874 &0.0396 &0.0161&\\
\hline
&\multirow{4}{*}{Uniform} & 0.01 & 787 &1.8034 &0.9996&196 &2.0296 &1&\\
&& 0.003 & 8815 &0.3621 &0.2791 &2203 &1.3586 &0.9447&\\
&& 0.002 & 19854 &0.0398 &0.03 &4963 &0.7823 &0.6043&\\
&& 0.001 & 79497 &2.63e-7 &0 &19874 &0.0396 &0.0275&\\
\hline
&\multirow{4}{*}{Laplace} & 0.01 & 787 &1.8034 &0.9952&196 &2.0296 &1&\\
&& 0.003 & 8815 &0.3621 &0.1835 &2203 &1.3586 &0.8466&\\
&& 0.002 & 19854 &0.0398 &0.0125 &4963 &0.7823 &0.4509&\\
&& 0.001 & 79497 &2.63e-7 &0 &19874 &0.0396 &0.0152&\\


\toprule [1 pt]

\end{tabular}
\end{table*}

To demonstrate the effectiveness of our theoretical results, we illustrate our proposed sampling number to fit a common and complex distribution, the Nakagami distribution. Nakagami-$m$ distribution provides good fitting to empirical multipath fading channel \cite{Nakagami}. The parameter $m$ in this part is $2$ and $\Omega=10$. Fig. \ref{fig:Number_normal1} shows the cumulative distribution function (CDF) of empirical distribution and real distribution. The simulated DMIM deviation $\varepsilon$ is $0.1$, $0.05$, $0.01$ and $0.001$. It is noted that the gap between the CDF of empirical distribution and that of real distribution is constrained by the DMIM deviation. Obviously, the gap decreases with the decreasing of the DMIM deviation. Particularly, the gap almost disappears when $\varepsilon=0.001$. In general, there is a tradeoff between the sampling number and the accuracy for empirical distribution, but DMIM can provide a new viewpoint on this by taking message importance into account.

\begin{figure}
  \centerline{\includegraphics[width=8.0cm]{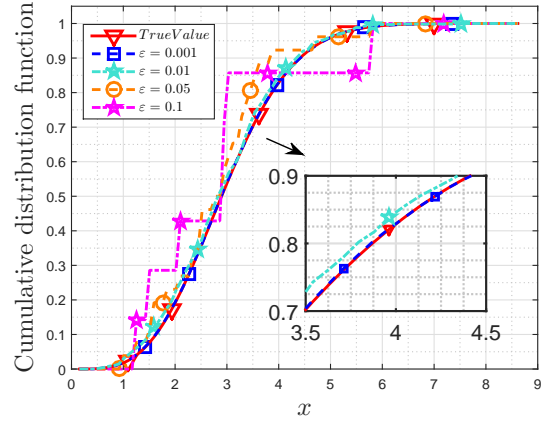}}
  \caption{The fitting of the cumulative distribution function of Nakagami distribution when $m=2$ and $\Omega=10$.}\label{fig:Number_normal1}
\end{figure}

\section{Conclusion}\label{sec:seven}
This paper focused on the problem, that how many samples is required in big data collection, with taking DMIM into account. Firstly, we defined DMIM as an measure of message importance for continuous random variable to help us describe the information flows during sampling. It is an extension of MIM and similar to differential entropy. Then, the DMIM for some common distributions, such as normal and uniform distribution, were discussed. Moreover, we made the asymptotic analysis of Gaussian distribution. As a result, high-precision approximate values for DMIM of normal distribution were respectively given when variance is extremely big or relatively small.

Then we proved that the divergence between the empirical distribution and the real distribution is controlled by the DMIM deviation, which shows the deviation of DMIM is equivalent to Kolmogorov-Smirnov statistic. In fact, compared with Kolmogorov-Smirnov test, the new method based on DMIM gives us another viewpoint of information collection because it visually shows the information flow with the increasing of sampling points, which helps us to design sampling strategy for the actual system of big data. Moreover, similar to Kolmogorov-Smirnov test, the sampling number in our method is distribution-free, which only depends on the DMIM deviation when the random variable is given.

Proposing the joint differential message importance measure and using it to design high-efficiency big data analytic system are of our future interests.

\appendices
\section{Proof of Theorem \ref{thm:normal small sigma}}\label{Appendices A}
\begin{proof}
For convenience, we might as well take
\begin{equation}
T(N) = {1 \over {\sqrt {2\pi n\left( {n + 1} \right)} }}{\left( {{e \over {\sqrt {2\pi } \sigma n}}} \right)^n},
\end{equation}
and let $n_0'=\left\lceil {{e \over {\sqrt {2\pi } \sigma }}} \right\rceil =c/\sigma+c'=n_0+1$ where $\left\lceil x\right\rceil$ is the smallest integer larger than or equal to $x$ and $c=e/\sqrt{2\pi}$. Obviously, $0\le c' \le 1$.

Hence
\begin{flalign}\label{equ:Normal small 00}
   \left| {l(X) - \hat{l}( X)} \right| &= \left| {\sum\limits_{n = {n_0}^\prime }^\infty  {{{{{\left( { - 1} \right)}^n}} \over {\sqrt {2\pi n\left( {n + 1} \right)} }}{{\left( {{e \over {\sqrt {2\pi } \sigma n}}} \right)}^n}} } \right|  \\
    &= \left| {\sum\limits_{n = {n_0}^\prime }^\infty  {{{\left( { - 1} \right)}^n}T(n)} } \right|  \tag{\theequation a} \label{equ:Normal small 00 a}\\
    &\le \sum\limits_{n = {n_0}^\prime }^\infty  {\left| {T(n)} \right|} .  \tag{\theequation b}\label{equ:Normal small 00 b}
\end{flalign}
This means, we only need to check $\sum\limits_{n = {n_0}^\prime }^\infty  {\left| {T(n)} \right|} <\frac{3\sigma}{e}$ holds.

Then we find
\begin{flalign}\label{equ:Normal small n0}
T(n_0')=&{1 \over {\sqrt {2\pi n_0'\left( {n_0'+ 1} \right)} }}{\left( {  {e \over {\sqrt {2\pi } \sigma n_0'}}} \right)^{n_0'}} \\
 =& {1 \over {\sqrt {2\pi \left( {c/\sigma  + c'} \right)\left( {c/\sigma  + c' + 1} \right)} }}{\left( {  {c \over {\sigma \left( {c/\sigma  + c'} \right)}}} \right)^{c/\sigma  + c'}}  \nonumber \\
  = & {1 \over {\sqrt {2\pi \left( {c/\sigma  + c'} \right)\left( {c/\sigma  + c' + 1} \right)} }}\nonumber\\
 & \quad\quad \cdot{\left( {{{\left( {1 + {{c'\sigma } \over c}} \right)}^{{c \over {c'\sigma }}}}} \right)^{ - c'}}{\left( {1 + {{c'\sigma } \over c}} \right)^{ - c'}} . \tag{\theequation a}\label{equ:Normal small n0 a}
\end{flalign}

When $N>n_0'$, we obtain
\begin{flalign}\label{equ:Normal small 01}
   T(N) &= {1 \over {\sqrt {2\pi N\left( {N + 1} \right)} }}{\left( {{e \over {\sqrt {2\pi } \sigma N}}} \right)^N}  \\
  &  < {1 \over {\sqrt {2\pi {{n'}_0}\left( {{{n}_0'} + 1} \right)} }}{\left( {{e \over {\sqrt {2\pi } \sigma N}}} \right)^N} \tag{\theequation a}\label{equ:Normal small 01 a} \\
  &  = {1 \over {\sqrt {2\pi {{n}_0'}\left( {{{n}_0'} + 1} \right)} }}{\left( {{e \over {\sqrt {2\pi } \sigma N}}} \right)^{{{n}_0'}}}{\left( {{e \over {\sqrt {2\pi } \sigma N}}} \right)^{N - {{n}_0'}}} \nonumber \\
  &  = {1 \over {\sqrt {2\pi {{n'}_0}\left( {{n_0'} + 1} \right)} }}{\left( {{e \over {\sqrt {2\pi } \sigma {n_0'}}}{{{n_0'}} \over N}} \right)^{{{n}_0'}}}{\left( {{e \over {\sqrt {2\pi } \sigma N}}} \right)^{N - {{n}_0'}}} \nonumber\\
  &  = {1 \over {\sqrt {2\pi {{n}_0'}\left( {{{n}_0'} + 1} \right)} }}{\left( {{e \over {\sqrt {2\pi } \sigma {{n}_0'}}}} \right)^{{{n'}_0}}}\nonumber \\
  &\quad \quad \quad \cdot {\left( {1 + {{N - {{n'}_0}} \over {{n_0'}}}} \right)^{ - {n_0'}}}{\left( {{e \over {\sqrt {2\pi } \sigma N}}} \right)^{N - {{n}_0'}}}  \tag{\theequation b}\label{equ:Normal small 01 b}\\
  &  < T({n_0'}){\left( {{e \over {\sqrt {2\pi } \sigma N}}} \right)^{N - {n_0'}}}  \tag{\theequation c}\label{equ:Normal small 01 c}\\
  &  \le \left\{
  \begin{aligned}
    & T({n_0'}){e \over {\sqrt {2\pi } \sigma }} \frac{1}{N}, \,\,\,\,\quad  N-n_0'=1 \\
   & T({n_0'}){\left({e \over {\sqrt {2\pi } \sigma}}\right)}^2 \frac{1}{ \left( {N - 1} \right)N},\,N-n_0'\ge 2 \\
  \end{aligned}
   \right.  ,\tag{\theequation d}\label{equ:Normal small 01 d}
\end{flalign}
where (\ref{equ:Normal small 01 a}) follows from ${1 \over {\sqrt {2\pi N\left( {N + 1} \right)} }} < {1 \over {\sqrt {2\pi n_0'\left( {n_0' + 1} \right)} }} $ because $N>n_0'$. (\ref{equ:Normal small 01 c}) is obtained by removing ${\left( {1 + {{N - {{n'}_0}} \over {{n_0'}}}} \right)^{ - {n_0'}}}$. It requires that $0<{\left( {1 + {{N - {{n'}_0}} \over {{n_0'}}}} \right)^{ - {n_0'}}}<1$. Such a condition is satisfied because $N>n_0'>0$. It is obtained that ${1 \over {{N^{N - {n_0}^\prime }}}} \le {1 \over {\left( {N - 1} \right)N}}$ when $N-n_0'\ge2$. Therefore (\ref{equ:Normal small 01 d}) holds.

Substituting (\ref{equ:Normal small 01 d}) in (\ref{equ:Normal small 00 b}), we have (\ref{equ:Normal small 02})-(\ref{equ:Normal small approx 1}) (See the nest page).
\begin{figure*}
\begin{flalign}\label{equ:Normal small 02}
  \sum\limits_{n = {n_0}^\prime }^\infty  {\left| {T(n)} \right|}&  < T({n_0}')\left( {1 + {e \over {\sqrt {2\pi } \sigma }}{1 \over {{n_0}^\prime  + 1}} + {e \over {\sqrt {2\pi } \sigma }}{1 \over {\left( {{n_0}^\prime  + 1} \right)\left( {{n_0}^\prime  + 2} \right)}} + {e \over {\sqrt {2\pi } \sigma }}{1 \over {\left( {{n_0}^\prime  + 2} \right)\left( {{n_0}^\prime  + 3} \right)}}+...} \right)  \\
  &  = T({n_0}')\left( {1 + {e \over {\sqrt {2\pi } \sigma }}{1 \over {{n_0}^\prime  + 1}} + {e \over {\sqrt {2\pi } \sigma }}{1 \over {{n_0}^\prime  + 1}} - {e \over {\sqrt {2\pi } \sigma }}{1 \over {{n_0}^\prime  + 2}} + {e \over {\sqrt {2\pi } \sigma }}{1 \over {{n_0}^\prime  + 2}}+...} \right) \tag{\theequation a} \label{equ:Normal small 02 a} \\
  &  = T({n_0}')\left( {1 + {{2c} \over \sigma }{1 \over {{n_0}^\prime  + 1}}} \right)  \tag{\theequation b}\label{equ:Normal small 02 b} \\
  &  = {1 \over {\sqrt {2\pi \left( {c/\sigma  + c'} \right)\left( {c/\sigma  + c' + 1} \right)} }}{\left( {{{\left( {1 + {{c'\sigma } \over c}} \right)}^{{c \over {c'\sigma }}}}} \right)^{ - c'}}{\left( {1 + {{c'\sigma } \over c}} \right)^{ - c'}}\left( {1 + {{2c} \over {c + c'\sigma  + \sigma }}} \right) . \tag{\theequation c}\label{equ:Normal small 02 c}
\end{flalign}
(\ref{equ:Normal small 02 c}) is obtained by substituting (\ref{equ:Normal small n0 a}) in (\ref{equ:Normal small 02 b}).
In fact, we find
\begin{equation}\label{equ:Normal small approx 1}
\mathop {\lim }\limits_{\sigma  \to 0} {{{1 \over {\sqrt {2\pi \left( {c/\sigma  + c'} \right)\left( {c/\sigma  + c' + 1} \right)} }}{{\left( {{{\left( {1 + {{c'\sigma } \over c}} \right)}^{{c \over {c'\sigma }}}}} \right)}^{ - c'}}{{\left( {1 + {{c'\sigma } \over c}} \right)}^{ - c'}}\left( {1 + {{2c} \over {c + c'\sigma  + \sigma }}} \right)} \over {3\sigma {e^{ - c' - 1}}}} = 1.
\end{equation}
Therefore, when $\sigma$ is relatively small, $3\sigma e^{-c'-1}$ is a good approximate value for (\ref{equ:Normal small 02 c}).
\end{figure*}
Based on the discussions above, when $\sigma$ is relatively small, we have
\begin{equation}\label{equ:Normal small result 1}
\sum\limits_{n = {n_0}^\prime }^\infty  {\left| {T(n)} \right|} <{{3\sigma } \over e}{e^{ - c'}}.
\end{equation}
In fact $0\le c'\le1$, so we obtain
\begin{equation}\label{equ:Normal small result 2}
\sum\limits_{n = {n_0}^\prime }^\infty  {\left| {T(n)} \right|} <{{3\sigma } \over e}.
\end{equation}
Hence,
\begin{equation}\label{equ:Normal small result 3}
\left| {l(x) - l(\hat x)} \right| < {{3\sigma } \over e}.
\end{equation}
The proof is completed.
\end{proof}

\section{Proof of Theorem \ref{thm:low bound number}}\label{Appendices B}
\begin{proof}
In fact, a upper bound of $ P\{D_n>d\}$ is given by
\begin{flalign}\label{equ: up_bound sum}
  & P\left\{ {{D_n} > d} \right\} \approx 2\sum\limits_{k = 1}^\infty  {{{\left( { - 1} \right)}^{k - 1}}{e^{ - 2n{k^2}{d^2}}}}   \\
  &  = 2\sum\limits_{m = 1}^\infty  {\left( {{e^{ - 2n{{\left( {2m - 1} \right)}^2}{d^2}}} - {e^{ - 2n{{(2m - 1 + 1)}^2}{d^2}}}} \right)}   \tag{\theequation a}\label{equ: up_bound sum a}\\
  &  = 2\sum\limits_{m = 1}^\infty  {\left( {{e^{ - 2n{{\left( {2m - 1} \right)}^2}{d^2}}}\left( {1 - {e^{ - 2n\left( {4m - 1} \right){d^2}}}} \right)} \right)}   \tag{\theequation b}\label{equ: up_bound sum b}\\
  &  \le 2\sum\limits_{m = 1}^\infty  {{e^{ - 2n{{\left( {2m - 1} \right)}^2}{d^2}}}}   \tag{\theequation c}\label{equ: up_bound sum c}\\
  &  \le 2\sum\limits_{m = 1}^\infty  {{e^{ - 4n{d^2}\left( {2m - 1} \right) + 2n{d^2}}}}   \tag{\theequation d}\label{equ: up_bound sum d}\\
  &  = 2\sum\limits_{m = 1}^\infty  {{e^{ - 8n{d^2}m + 6n{d^2}}}}    \tag{\theequation e}\label{equ: up_bound sum e}\\
  &  = {{2{e^{ - 2n{d^2}}}} \over {1 - {e^{ - 8n{d^2}}}}}  .\tag{\theequation f}\label{equ: up_bound sum f}
\end{flalign}
(\ref{equ: up_bound sum c}) is obtained for the fact that $1 - {e^{2n\left( {4m - 1} \right){d^2}}} \leq 1$. (\ref{equ: up_bound sum d}) requires $ - 2n{\left( {2m - 1} \right)^2}{d^2} \le  - 4n{d^2}\left( {2m - 1} \right) + 2n{d^2}$. Such a condition is satisfied because $ - 2n{d^2}{\left( {2m - 1 - 1} \right)^2} \le 0$.

This means, we only need to check ${{{e^{ - 2n{d^2}}}} \over {1 - {e^{ - 8n{d^2}}}}} \le \beta $ holds. 

Substituting (\ref{equ:XNX define 1}) and (\ref{equ:XNX infty}) in $\left| \gamma(\infty)-\gamma(n)\right| \leq \varepsilon$, we get
\begin{equation}\label{equ:NMIM_number_1}
   \left| {{1 \over {l\left( X \right)}} - {{{e^{ - {1 \over {2\sqrt {\pi n} \sigma }}}}} \over {l\left( X \right)}}} \right| \le \varepsilon \Rightarrow n \ge {1 \over {4\pi {\sigma ^2}{{\ln }^2}\left( {1 - \varepsilon l\left( X \right)} \right)}}.
\end{equation}

Because $0\leq l(X) \leq 1$, we obtain
\begin{equation}\label{equ:number_choose}
n \ge {1 \over {4\pi {\sigma ^2}{{\ln }^2}\left( {1 - \varepsilon l\left( X \right)} \right)}} \ge {1 \over {4\pi {\sigma ^2}{{\ln }^2}\left( {1 - \varepsilon } \right)}}.
\end{equation}
Letting
\begin{equation}\label{equ:NMIM_proof d}
d=\sqrt{2\pi\sigma^2 \ln{\frac{19}{9\beta}}} \ln{\frac{1}{1-\varepsilon}},
\end{equation}
we have
 \begin{flalign}\label{equ:NMIM_proof 1}
   2n{d^2} &\ge 2{{2\pi {\sigma ^2}\ln {{19} \over {9\beta }}{{\ln }^2}(1 - \varepsilon )} \over {4\pi {\sigma ^2}{{\ln }^2}(1 - \varepsilon )}}   \Rightarrow  {e^{ - 2n{d^2}}} \le {{9\beta } \over {19}}.
\end{flalign}
It is easy to check
\begin{equation}\label{equ:NMIM_proof 2}
\beta {\left( {{e^{ - 2n{d^2}}}} \right)^4} +2 {e^{ - 2n{d^2}}} - \beta  \le 0,
\end{equation}
when $\beta \le {{19} \over 9}\root 4 \of {{1 \over {19}}} \approx1.0112$. In fact, $\beta$ is a threshold value of the probability, so we usually take $\beta \le 1$. Therefore, (\ref{equ:NMIM_proof 2}) holds all the time.

Hence,
\begin{equation}
{2{{e^{ - 2n{d^2}}}} \over {1 - {e^{ - 8n{d^2}}}}} \le \beta .
\end{equation}
Based on the discussions above, we get
  \begin{equation}
   P\left\{D_n>\sqrt{2\pi\sigma^2 \ln{\frac{19}{9\beta}}} \ln{\frac{1}{1-\varepsilon}}\right\}< \beta.
  \end{equation}
\end{proof}

\end{document}